\documentclass[10pt]{amsart}
\usepackage{amsmath}
\usepackage{latexsym}
\usepackage{amsfonts}
\usepackage{amssymb}
\usepackage{color}
\usepackage{bbm,dsfont}

\newtheorem{proposition}{Proposition}
\newtheorem{theorem}{Theorem}
\newtheorem{lemma}{Lemma}
\newtheorem{corollary}{Corollary}

\theoremstyle{definition}

\newtheorem{definition}{Definition}

\newcommand{\frecc}{\longrightarrow} 

\newcommand{\de}{{\rm d}} 

\newcommand{\lft}{\left(} 
\newcommand{\rgt}{\right)}

\newcommand{\real}{\mathbb R} 
\newcommand{\torus}{\mathbb{T}} 
\newcommand{\complex}{{\mathbb C}} 
\newcommand{\nat}{\mathbb N} 

\newcommand{\hotimes}{\hat{\otimes}} 
\renewcommand{\aa}{\mathcal{A}} 
\newcommand{\pp}{\mathcal{P}}
\newcommand{\hh}{\mathcal{H}} 
\newcommand{\kk}{\mathcal{K}} 
\newcommand{\vv}{\mathcal{V}} 
\newcommand{\xx}{\mathcal{X}} 
\newcommand{\Z}{\mathbbm{Z}}
\newcommand{\cc}{\mathcal{C}} 
\newcommand{\ttt}{\mathcal{T}}
\renewcommand{\lll}{\mathcal{L}}
\newcommand{\trh}{\mathcal{T}(\mathcal{H})} 
\newcommand{\trhzero}{\mathcal{T}_0(\mathcal{H})} 
\newcommand{\lh}{\mathcal{L(H)}} 
\newcommand{\elle}[1]{\mathcal{L} ( #1 )} 
\newcommand{\ti}[1]{\mathcal{T} \left( #1 \right)} 
\newcommand{\spanno}[1]{{\rm span}\, \left\{ #1 \right\}} 
\newcommand{\rad}[1]{{\rm rad}\, #1} 

\newcommand{\eps}{\varepsilon}
\newcommand{\scal}[2]{\left\langle #1 , #2 \right\rangle} 
\newcommand{\norm}[1]{\left\|#1\right\|} 
\newcommand{\dom}{{\rm dom}\,} 
\newcommand{\tr}[1]{{\rm tr} \left[ #1 \right]} 
\newcommand{\trap}[2]{{\rm tr}^{#1} \left[ #2 \right]} 
\newcommand{\trped}[2]{{\rm tr}_{#1} \left[ #2 \right]} 

\newcommand{\ii}{\mathcal{I}} 
\newcommand{\jj}{\mathcal{J}} 
\newcommand{\memo}{\mathfrak{M}} 

\newcommand{\mm}{\mathcal{M}} 

\newcommand{\bor}[1]{\mathcal{B}\lft{#1}\rgt} 
\newcommand{\boro}{\bor{\Omega}} 
\newcommand{\borel}[1]{{\mathcal B}\lft \real\rgt}

\newcommand{\sss}{\mathcal{S}} 
\newcommand{\ldueG}{L^2 (G)}

\newcommand{\ldueGV}{L^2 (G; \mathcal{V})}
\newcommand{\ldueGVhh}{L^2 (G; \mathcal{V} \otimes \mathcal{H})}

\begin{document}
\setlength\arraycolsep{2pt}
\title[Covariant quantum instruments]{Covariant quantum instruments}

\begin{abstract}
The structure of covariant instruments is
studied and a general structure theorem is derived. A detailed characterization is given to covariant instruments in the case of an irreducible representation of a locally compact group.
\end{abstract}

\author[Carmeli]{Claudio Carmeli}
\address{Claudio Carmeli, Dipartimento di Fisica, Universit\`a di Genova, and I.N.F.N., Sezione di Genova, Via
Dodecaneso 33, 16146 Genova, Italy}
\email{carmeli@ge.infn.it}

\author[Heinosaari]{Teiko Heinosaari}
\address{Teiko Heinosaari, Research Center for Quantum Information, Slovak Academy of Sciences, Bratislava, Slovakia and Department of Physics and Astronomy, University of Turku, Finland}
\email{teiko.heinosaari@utu.fi}

\author[Toigo]{Alessandro Toigo}
\address{Alessandro Toigo, Dipartimento di Informatica, Universit\`a di Genova, and I.N.F.N., Sezione di Genova, Via
Dodecaneso 35, 16146 Genova, Italy}
\email{toigo@ge.infn.it}

\date{\the\day/\the\month/\the\year}

\maketitle{}

\section{Introduction}\label{sec:introduzione}

An instrument captures neatly the mathematical description of a quantum measurement. For each input state, the instrument gives both the measurement outcome probabilities and the conditional output states. The concept of an instrument was introduce by Davies and Lewis in \cite{DaLe70} and it has become a standard tool in quantum information theory \cite{Werner01},\cite{Keyl02} and in studies of various aspects of quantum measurements \cite{QTOS76},\cite{PSAQT82},\cite{Ozawa84},\cite{QTM96}. 

In this work we investigate the mathematical structure of covariant instruments. Covariance of an instrument means that there is a group having both a continuous unitary representation and a continuous action on outcome space, and that the instrument transforms in a consistent way under these operations. The covariance property is typical for instruments arising from physical applications. 

Covariant instruments were first studied by Davies in \cite{Davies70}, where he characterized their structure in the case of a compact group having a finite dimensional unitary representation. In \cite{Holevo98} Holevo investigated the structure of covariant instruments in the situation of a locally compact Abelian group. In this work we focus on the case of an irreducible representation of a locally compact group.

Our investigation proceeds in the following way. In Section \ref{sec:definizioni} we fix the notation and recall the definition of a covariant instrument. In Section \ref{sec:struttura} we derive a general structure theorem for covariant instruments. This theorem shows that a covariant instrument is determined by a system of imprimitivity and an intertwining operator. Section \ref{sec:irriducibile} concentrates on the case of an irreducible representation and a transitive action with a compact stability subgroup. For this kind of situation we derive a characterization of all covariant instruments. In Section \ref{sec:projective} these results are generalized to cover the case of a projective unitary representation since this is the general framework in quantum mechanics. Finally, in Section \ref{sec:cp} we give an alternative formulation for the characterizations obtained in Sections \ref{sec:irriducibile} and \ref{sec:projective}.

\section{Basic definitions}\label{sec:definizioni}

If $\xx_1$ and $\xx_2$ are Banach spaces, we denote by $\elle{\xx_1 ; \xx_2}$ the Banach space of the bounded operators from $\xx_1$ to $\xx_2$, with the uniform norm $\norm{\cdot}_\infty$. We also use abbreviate notation $\elle{\xx ; \xx} = \elle{\xx}$.

Let $\hh$ be a complex separable Hilbert space. We denote by $\norm{\cdot}_\hh$ its norm and $\scal{\cdot}{\cdot}_\hh$ its scalar product, assumed linear in the first entry. (When no confusion will arise, the subscripts are dropped.) Let $\lh$ and $\trh$ be the Banach spaces of bounded operators and trace class operators on $\hh$, respectively. We denote by $\norm{\cdot}_\lll$ the operator norm on $\lh$ and $\norm{\cdot}_{\ttt}$ the trace class norm on $\trh$. For each $u,v\in\hh$, we denote by $u\odot v$ the rank one operator on $\hh$ defined as
\begin{equation*}
(u\odot v) (w) = \scal{w}{v} u \quad \forall w\in\hh.
\end{equation*}

Let $\Omega$ be a locally compact topological space, which is Hausdorff and satisfies the second axiom of countability (lcsc space, in short). We let $\bor{\Omega}$ denote the Borel $\sigma$-algebra of $\Omega$.

An instrument has several equivalent definitions. Often an instrument is defined as a $\sigma$-additive map $\ii$ from $\boro$ to the set $\elle{\trh}$ of bounded linear maps on $\trh$. It is then required that $\ii(X)$ is a completely positive map for each $X\in\boro$, and that $\ii$ satisfies the normalization condition $\tr{\ii(\Omega)T}=\tr{T}$ for each $T\in\trh$.  In our current investigation it is more convenient to use a slightly different but equivalent definition for instruments. For this purpose, let $\mm(\Omega ; \trh)$ be the ordered Banach space of $\trh$-valued Borel measures on $\Omega$, with norm $\norm{M}_\mm = |M| (\Omega)$, $|M|$ being the total variation of $M$; see, for instance, \cite{RFA93}. An instrument can now be seen as a map from $\trh$ to  $\mm(\Omega ; \trh)$. In the following we state this alternative definition explicitly.

\begin{definition}\label{def:str}
An \emph{instrument} based on $\Omega$ is a linear map $\ii : \trh \frecc \mm(\Omega ; \trh)$ such that
\begin{itemize}
\item[({\rm i})] for each $X \in \bor{\Omega}$, the linear map
\begin{equation*}
\ii_X : \trh \frecc \trh,\quad T\mapsto \ii_X (T):= (\ii T)(X)
\end{equation*}
is completely positive;
\item[({\rm ii})] for every $T\in\trh$,
\begin{equation*}
\tr{(\ii T)(\Omega)} = \tr{T}.
\end{equation*}
\end{itemize}
\end{definition}

We recall (as we will need these notations later) that the complete positivity of a map $\ii_X$ means the following. If $n\in\nat$, let $\hh^{(n)}$ be the direct sum of $n$ copies of $\hh$. We identify $\hh^{(n)}$ with the column vectors having $n$ entries in $\hh$. In this way, each trace class operator $\tilde{T} \in \ti{\hh^{(n)}}$ is identified with an $n \times n$ matrix with entries in $\trh$. Let $\ii^{(n)}_X (\tilde{T})$ be the element of $\ti{\hh^{(n)}}$ defined as $[\ii^{(n)}_X (\tilde{T})]_{ij} = \ii_X ([\tilde{T}]_{ij})$. To require complete positivity of $\ii_X$ is to say that for each $n$, the mapping $\ii^{(n)}_X : \ti{\hh^{(n)}} \frecc \ti{\hh^{(n)}}$ is positive.

\begin{proposition}\label{prop:bounded}
An instrument $\ii$ is a bounded map and $\norm{\ii}_{\infty} \leq 2$.
\end{proposition}

\begin{proof}
By condition ({\rm i}) of Definition \ref{def:str}, each map $\ii_X$ is, in particular, positive. This implies that $\ii$ is a positive map. Let $T\in\trh$. We can decompose $T$ into a sum $T = T^1_+ - T^1_- + i (T^2_+ - T^2_-)$, where $T^i_{\pm} \geq 0$ and $\norm{T^i_+}_{\ttt} + \norm{T^i_-}_{\ttt} \leq \norm{T}_{\ttt}$. This decomposition and condition ({\rm ii}) of  Definition \ref{def:str} imply that
\begin{eqnarray*}
\norm{\ii T}_{\mm} &\leq& \sum_{i = 1,2,\, j = +,-} \norm{\ii T^i_j}_{\mm} = \sum_{i = 1,2,\, j = +,-} \tr{(\ii T^i_j)(\Omega)} \\
& = & \sum_{i = 1,2,\, j = +,-} \tr{T^i_j}
\leq 2 \norm{T}_{\ttt}.
\end{eqnarray*}
\end{proof}

Let $G$ be a lcsc~topological group having a strongly continuous unitary representation $g\mapsto U(g)\equiv U_g$ on $\hh$, and acting continuously on $\Omega$. The latter requirement means that there exists a continuous mapping $G\times \Omega \ni (g,\omega)\mapsto g\cdot\omega \in \Omega$ such that
\begin{itemize}
\item the mapping $\omega\mapsto g\cdot\omega$ is a homeomorphism of $\Omega$ for each $g\in G$;
\item $g_1\cdot (g_2\cdot\omega)=(g_1g_2)\cdot\omega$ for every $g_1,g_2\in G$ and $\omega\in\Omega$ .
\end{itemize}
If $g\in G$ and $X\subseteq\Omega$, we denote $g\cdot X=\{g\cdot x\mid x\in X\}$.

\begin{definition}\label{def:cov}
An instrument $\ii$ is \emph{covariant with respect to $U$}, or shortly \emph{$U$-covariant}, if
\begin{equation}\label{covar}
\ii_{g\cdot X}(T) = U_g\ii_X ( U^\ast_g T U_g )U_g^\ast   \quad \forall X\in \bor{\Omega},\, g\in G,\, T\in \trh.
\end{equation}
\end{definition}

\section{General structure theorem}\label{sec:struttura}

In his seminal article \cite{Ozawa84} Ozawa presented a fundamental structure theorem for instruments. Theorem \ref{teo:struttura} and Corollary \ref{cor:structure} are generalizations of Ozawa's result taking into account the covariance property of an instrument\footnote{The action of $G$ is not required to be transitive. Therefore, any instrument is covariant if $G$ is chosen to be the trivial group of one element. In this way Corollary 5.2 of \cite{Ozawa84} is contained in Corollary \ref{cor:structure}.}. 
A similar result has been stated in \cite{Holevo98} and proved in \cite{Denisov90} in a slightly weaker form (i.e. under the hypothesis that $G$ acts transitively on $\Omega$ and without proving separability of the auxiliary Hilbert space $\mathcal{K}$).
Theorem \ref{teo:struttura} will play a crucial role in our investigation in Section \ref{sec:irriducibile} so we find it useful to give a detailed proof here. 

To formulate Theorem \ref{teo:struttura} and its proof we need to fix the following notation. Let $\kk$ be a Hilbert space. We denote by $\kk\otimes\hh$ the Hilbert space tensor product of $\kk$ and $\hh$. The partial trace over $\kk$ is the linear map ${\rm tr}^\kk : \ti{\kk \otimes \hh} \frecc \trh$ defined by the condition
\begin{equation*}
\tr{A\,\trap{\kk}{T}} = \tr{\lft I\otimes A \rgt T} \quad \forall T\in\ti{\kk\otimes\hh},\, A\in\lh,
\end{equation*}
where the trace on the left-hand side is over $\hh$ and on the right-hand side over $\kk\otimes\hh$.

\begin{theorem}\label{teo:struttura}
Let $\ii$ be a $U$-covariant instrument. Then there exist
\begin{itemize}
\item a separable Hilbert space $\kk$, a strongly continuous unitary representation $D$ of $G$ in $\kk$ and a projection valued measure $P:\bor{\Omega} \frecc \elle{\kk}$ satisfying
\begin{equation}\label{eq:struttura:P}
P (g\cdot X) = D_g P(X) D^\ast_g \quad \forall X\in\bor{\Omega},\, g\in G;
\end{equation}
\item an isometry $W:\hh \frecc \kk\otimes\hh$ satisfying
\begin{equation}\label{eq:struttura:W}
W U_g = \lft D_g\otimes U_g \rgt W \quad \forall g\in G ,
\end{equation}
\end{itemize}
such that
\begin{equation}\label{eq:struttura:I}
(\ii T)(X) = \trap{\kk}{\lft P(X) \otimes I \rgt W T W^\ast}.
\end{equation}
Moreover, $\kk$, $D$, $P$ and $W$ can be chosen in such a way that the set 
\begin{equation}\label{eq:struttura:set}
\left\{ (P(X)\otimes A) Wv \mid X\in \bor{\Omega},\, A\in \lh,\, v\in\hh \right\}
\end{equation}
is total in $\kk \otimes \hh$. This requirement makes the imprimitivity system $(D,P,\kk)$ unique up to an isomorphism, i.e., if $\kk^\prime$, $D^\prime$, $P^\prime$, $W^\prime$ are respectively as $\kk$, $D$, $P$, $W$ , then there exists a unitary map $V:\kk\frecc\kk^\prime$ such that $VD = D^\prime V$ and $VP = P^\prime V$.
\end{theorem}

\begin{proof}
For each $T\in \trh$ and $A\in\lh$, we denote by $\mu_{A;T}$ the complex Borel measure defined by
\begin{equation*}
\mu_{A;T}(X) := \tr{A \ii_X (T)} \quad \forall X\in\bor{\Omega}.
\end{equation*}

We divide the proof into steps (A)-(G). 
\begin{itemize}
\item[(A)] For each set $X\in \Omega$, denote by $\chi_X$ the characteristic function of $X$. Define $\sss (\Omega) := \spanno{\chi_X \mid X\in\bor{\Omega}}$, a subset of the space of the Borel functions on $\Omega$, and let $\hat{\hh}_0 := \sss (\Omega) \hotimes \lh \hotimes \hh$, where $\hotimes$ denotes algebraic tensor product.
The following map from $\left[ \sss (\Omega ) \times \lh \times \hh \right]^2$ into $\complex$
\begin{equation*}
(f_1,A_1,v_1 ; f_2,A_2,v_2) \mapsto \int f_1 (x) \overline{f}_2 (x) \de\mu_{A_2^\ast A_1 ; v_1 \odot v_2 } (x)
\end{equation*}
defines a sesquilinear form $\scal{\cdot}{\cdot}_0$ on $\hat{\hh}_0$. This form is positive semidefinite. In fact, if $\phi\in \hat{\hh}_0$, there exist disjoint sets $X_1 , X_2\ldots X_n$ in $\bor{\Omega}$ and, for each $i=1,2\ldots n$, elements $A^{(i)}_1 , A^{(i)}_2 \ldots A^{(i)}_m \in \lh$ and $v^{(i)}_1 , v^{(i)}_2 \ldots v^{(i)}_m \in \hh$ such that
\begin{equation}\label{decomp. di fi in k0}
\phi = \sum_{i=1}^n \sum_{j=1}^m \chi_{X_i} \hotimes A^{(i)}_j \hotimes v^{(i)}_j.
\end{equation}
Let $\tilde{A}^{(i)}$ be the matrix in $\elle{\hh^{(m)}}$ with entries $\tilde{A}^{(i)}_{hk} = \delta_{1h} A^{(i)}_k$, and let $\tilde{v}^{(i)}$ be the vector in $\hh^{(m)}$ with $\tilde{v}^{(i)}_h = v^{(i)}_h$. We have
\begin{eqnarray*}
\scal{\phi}{\phi}_0 & = & \sum_{i=1}^n \sum_{j,j^\prime = 1}^m \trped{\hh}{A_j^{(i)} \ii_{X_i} (v_j^{(i)} \odot v_{j^\prime}^{(i)}) A_{j^\prime}^{(i) \ast}} \\
& = & \sum_{i=1}^n \sum_{j,j^\prime = 1}^m\sum_{h = 1}^m \trped{\hh}{\tilde{A}_{hj}^{(i)} \ii_{X_i}^{(m)} (\tilde{v}^{(i)} \odot \tilde{v}^{(i)})_{j j^\prime} \tilde{A}_{h j^\prime}^{(i) \ast}} \\
& = & \sum_{i=1}^n \sum_{h = 1}^m \trped{\hh}{ ( \tilde{A}^{(i)} \ii_{X_i}^{(m)} (\tilde{v}^{(i)} \odot \tilde{v}^{(i)}) \tilde{A}^{(i) \ast} )_{hh} } \\
& = & \sum_{i=1}^n \trped{\hh^{(m)}}{ \tilde{A}^{(i)} \ii_{X_i}^{(m)} (\tilde{v}^{(i)} \odot \tilde{v}^{(i)}) \tilde{A}^{(i) \ast} }
\end{eqnarray*}
and therefore $\scal{\phi}{\phi}_0\geq 0$ by the positivity of $\ii^{(m)}_{X_i}$. Hence, denoting by $\rad{\scal{\cdot}{\cdot}_0}$ the kernel of the map $\phi\mapsto\scal{\phi}{\phi}_0$, then the quotient space $\hat{\hh}_0 / \rad{\scal{\cdot}{\cdot}_0}$ is a scalar product space in the usual way. We denote by $\hat{\hh}$ the Hilbert space obtained by completing this quotient space.

\item[(B)] We show that the Hilbert space $\hat{\hh}$ constructed in (A) is separable. Since $\hh$ is separable, there exist a sequence $\{ v_n \}_{n\in\nat}$ which is dense in $\hh$ and a sequence $\{ A_n \}_{n\in\nat}$ which is dense in $\lh$ with the ultra-strong (i.e. $\sigma$-strong) operator topology. Moreover, since $\Omega$ is second countable, there exists a sequence $\{ X_n \}_{n\in\nat}$ in $\bor{\Omega}$ with the following property: if $\mu$ is a positive measure on $\Omega$ and $X\in\bor{\Omega}$, for every $\eps > 0$ there exists $n$ such that $\mu (X \Delta X_n) < \eps$ (here $\Delta$ denotes the symmetric difference between two sets; the claim follows from Theorem C Sect.~5 and Theorem D Sect.~13 in \cite{MT50}). In the following we show that the set $\spanno{\chi_{X_k} \hotimes A_j \hotimes v_i \mid k,j,i\in\nat }$ is dense in $\hat{\hh}_0$. By the density of $\hat{\hh}_0 / \rad{\scal{\cdot}{\cdot}_0}$ in $\hat{\hh}$, the separability of $\hat{\hh}$ then follows. 

Let $\eps>0$ and $\phi = \chi_{X} \hotimes A \hotimes v$. Choose $i$ such that 
$$
\tr{A^\ast A \ii_X ((v - v_i) \odot (v - v_i))} < \eps^2 /9,
$$
then choose $j$ such that
$$
\tr{(A - A_j)^\ast (A - A_j) \ii_X (v_i \odot v_i)} < \eps^2 /9
$$
and finally $k$ such that 
$$
\mu_{A_j^\ast A_j ; v_i \odot v_i} (X \Delta X_k) < \eps^2 /9.
$$
We then have
\begin{eqnarray*}
&&\norm{\phi - \chi_{X_k} \hotimes A_j \hotimes v_i}_{\hat{\hh}} \leq \norm{\chi_{X} \hotimes A \hotimes (v - v_i)}_{\hat{\hh}} \\
&&\qquad \qquad + \norm{\chi_{X} \hotimes (A - A_j) \hotimes v_i}_{\hat{\hh}} + \norm{\chi_{X\Delta X_k} \hotimes A_j \hotimes v_i}_{\hat{\hh}} < \eps
\end{eqnarray*}
and our claim is therefore proven.

\item[(C)] We now construct a projection valued measure on $\hat{\hh}$.  For each $X\in\bor{\Omega}$, we define the operator $\hat{P}(X)$ on $\hat{\hh}_0$ by formula
\begin{equation*}
\hat{P} (X) \lft f \hotimes A \hotimes v \rgt = \chi_{X} f\hotimes A \hotimes v .
\end{equation*}
If $\phi = \sum_{i = 1}^{n} \chi_{X_i} \hotimes A_i \hotimes v_i$ is an element in $\hat{\hh}_0$, we have
\begin{eqnarray*}
\scal{\hat{P} (X)\phi}{\hat{P} (X)\phi}_0 &=& \sum_{i,j=1}^n \tr{A^\ast_j A_i \ii_{X_i \cap X_j \cap X} (v_i \odot v_j)} \\
& = & \scal{\hat{P} (X)\phi}{\phi}_0 \leq \scal{\hat{P} (X)\phi}{\hat{P} (X)\phi}_0^{1/2} \scal{\phi}{\phi}_0^{1/2}
\end{eqnarray*}
by Cauchy-Schwartz inequality.
From this we deduce that $\hat{P} (X) \rad{\scal{\cdot}{\cdot}_0} \subset \rad{\scal{\cdot}{\cdot}_0}$, so that $\hat{P} (X)$ descends to the quotient space $\hat{\hh}_0 / \rad{\scal{\cdot}{\cdot}_0}$. Moreover, the previous calculation shows that $\hat{P} (X)$ extends to a bounded selfadjoint operator on $\hat{\hh}$. Clearly, $\hat{P} (X)^2 = \hat{P} (X)$.

We show that the mapping $X\mapsto \hat{P} (X)$ from $\bor{\Omega}$ into $\elle{\hat{\hh}}$ is weakly $\sigma$-additive. Since the range of $\hat{P}$ in $\elle{\hat{\hh}}$ is norm bounded and the set $\hat{\hh}_0 / \rad{\scal{\cdot}{\cdot}_0}$ is dense in $\hat{\hh}$, it suffices to show that $\scal{\hat{P} (\cup_k X_k) \phi}{\phi}_0 = \sum_k \scal{\hat{P} (X_k) \phi}{\phi}_0 $ for all $\phi\in\hat{\hh}_0$ and for all disjoint sequences $\{ X_k \}_{k\in\nat}$ in $\bor{\Omega}$. If $\phi$ is as before, we have
\begin{equation*}
\scal{\hat{P} (X)\phi}{\phi}_0 = \sum_{i,j=1}^n \mu_{A^\ast_j A_i ; v_i \odot v_j} (X_i \cap X_j \cap X),
\end{equation*}
and the claim follows from $\sigma$-additivity of $\mu_{A^\ast_j A_i ; v_i \odot v_j}$.

\item[(D)] In the following we construct a unitary representation of $G$ which forms an imprimitivity system with $\hat{P}$. For each $g\in G$, we introduce in $\hat{\hh}_0$ the linear operator $\hat{D}_g$ whose action on decomposable elements is
\begin{equation*}
\hat{D}_g \lft f \hotimes A \hotimes v \rgt = g\cdot f \hotimes A U_g^{\ast} \hotimes U_g v,
\end{equation*}
where $g\cdot f (x) = f(g^{-1} \cdot x)$. The $U$-covariance of $\ii$ then implies that
\begin{equation*}
\scal{\hat{D}_g \lft f \hotimes A \hotimes v \rgt}{\hat{D}_g \lft f \hotimes A \hotimes v \rgt}_0 = \scal{f \hotimes A \hotimes v}{f \hotimes A \hotimes v}_0.
\end{equation*}
We conclude that $\hat{D}_g$ defines an isometric operator in $\hat{\hh}$.

Since $\hat{D}_{g_1 g_2} = \hat{D}_{g_1} \hat{D}_{g_2}$, we see that $\hat{D}$ is a group homomorphism of $G$ into the unitary group of $\hat{\hh}$. It is a straightforward consequence of the definitions of $\hat{P}$ and $\hat{D}$ that
\begin{equation}\label{covar. nella dim.}
\hat{D}_g \hat{P} (X) \hat{D}_g^{\ast} = \hat{P}(g\cdot X) \quad \forall X\in\bor{\Omega},\, g\in G.
\end{equation}

We now show that $\hat{D}$ is weakly (hence strongly) continuous. By 22.20, item (b) in \cite{AHAI63}, it suffices to show that the map $g\mapsto \scal{\hat{D}_g \phi}{\phi^\prime}_{\hat{\hh}}$ is $\mu_G$-measurable for all $\phi , \phi^\prime\in\hat{\hh}$, where we denoted by $\mu_G$ the left-invariant Haar measure on $G$. By density of  $\hat{\hh}_0 / \rad{\scal{\cdot}{\cdot}_0}$ in $\hat{\hh}$, it is enough to show $\mu_G$-measurability of the maps $g\mapsto \scal{\hat{D}_g \phi}{\phi^\prime}_0$ for $\phi = \chi_X \hotimes A \hotimes v$ and $\phi^\prime = \chi_{X^\prime} \hotimes A^\prime \hotimes v^\prime$. We have
\begin{equation*}
\scal{\hat{D}_g \phi}{\phi^\prime}_0 =  \tr{A^{\prime\ast} A U_g^\ast \ii_{g\cdot X \cap X^\prime} (U_g v\odot v^\prime)},
\end{equation*}
and, if $\{ e_i \}_{i\in\nat}$ is a Hilbert basis of $\hh$, this equation can be written in the form 
\begin{equation*}
\scal{\hat{D}_g \phi}{\phi^\prime}_0 = \sum_i \sum_j \sum_k \scal{U_g v}{e_j} \scal{A^{\prime\ast} A U_g^\ast e_k}{e_i} \scal{\ii_{g\cdot X \cap X^\prime} (e_j \odot v^\prime) e_i}{e_k}.
\end{equation*}
Since the maps $g\mapsto \scal{U_g v}{e_j}$ and $g\mapsto \scal{A^{\prime\ast} A U_g^\ast e_k}{e_i}$ are Borel (actually continuous), it suffices to show that $g\mapsto \scal{\ii_{g\cdot X \cap X^\prime} (e_j \odot v^\prime) e_i}{e_k}$ is $\mu_G$-measurable. The set $E := \{ (g,g\cdot x) \mid g\in G,x\in X  \}$ is a Borel subset of $G\times\Omega$. Moreover, denoting by $E^g$ the section of $E$ at $g$, we have $E^g = g\cdot X$. Thus,
\begin{equation*}
\scal{\ii_{g\cdot X \cap X^\prime} (e_j \odot v^\prime) e_i}{e_k} =
\mu_{e_i \odot e_k ; e_j \odot v^\prime} (E^g \cap X^\prime),
\end{equation*}
and, since the map $g\mapsto \mu_{e_i \odot e_k ; e_j \odot v^\prime} (E^g \cap X^\prime)$ is $\mu_G$-measurable by Fubini theorem, the claim follows.

\item[(E)] For each $B\in\lh$, let $\pi (B) : \hat{\hh}_0 \frecc \hat{\hh}_0$ be the linear operator extending the following action on decomposable vectors:
\begin{equation*}
\pi (B) \lft f \hotimes A \hotimes v \rgt = f \hotimes B A \hotimes v.
\end{equation*}
If $\phi\in\hat{\hh}_0$ is written as in eq.~(\ref{decomp. di fi in k0}), then we get
\begin{equation*}
\scal{\pi (B) \phi}{\pi (B) \phi}_0 = \sum_{i=1}^n \trped{\hh^{(m)}}{\tilde{A}^{(i) \ast} \tilde{B}^\ast \tilde{B} \tilde{A}^{(i)} \ii_{X_i}^{(m)} (\tilde{v}^{(i)} \odot \tilde{v}^{(i)})},
\end{equation*}
where $\tilde{A}^{(i)}$ and $\tilde{v}^{(i)}$ are defined as in the step (A) of the proof, and $\tilde{B}$ is the matrix in $\elle{\hh^{(m)}}$ with $\tilde{B}_{hk} = \delta_{hk} B$. Since the operator $\ii_{X_i}^{(m)} (\tilde{v}^{(i)} \odot \tilde{v}^{(i)})$ is positive and $\tilde{A}^{(i) \ast} \tilde{B}^\ast \tilde{B} \tilde{A}^{(i)} \leq \norm{B}^2_\lll \tilde{A}^{(i) \ast} \tilde{A}^{(i)}$, we get
\begin{eqnarray*}
\scal{\pi (B) \phi}{\pi (B) \phi}_0
& \leq & \norm{B}^2_\lll \sum_{i=1}^n \trped{\hh^{(m)}}{\tilde{A}^{(i) \ast} \tilde{A}^{(i)} \ii_{X_i}^{(m)} (\tilde{v}^{(i)} \odot \tilde{v}^{(i)})} \\
& = & \norm{B}^2_\lll \scal{\phi}{\phi}_0
\end{eqnarray*}
This shows that $\pi (B)$ descends to the quotient space $\hat{\hh}_0 / \rad{\scal{\cdot}{\cdot}_0}$, and extends to a bounded operator in $\hat{\hh}$.

It is a straightforward consequence of the definitions of $\pi,\hat{P}$ and $\hat{D}$ that  $\pi (A) \hat{P} (X) = \hat{P} (X) \pi (A)$ and $\pi (A) \hat{D}_g = \hat{D}_g \pi (A)$ for all $A,X$ and $g$.

It is easy to check that $\pi$ is a $\ast$-homomorphism of $\lh$ in $\hat{\hh}$. We claim that it is normal. In fact, if $B_n \downarrow O$ in $\lh$, then
\begin{equation*}
\scal{\pi (B_n) \lft \chi_{X_1} \hotimes A_1 \hotimes v_1 \rgt}{\chi_{X_2} \hotimes A_2 \hotimes v_2}_0 = \tr{A_2^\ast B_n A_1 \ii_{X_1 \cap X_2} (v_1 \odot v_2)} \end{equation*}
and the right-hand side goes to $0$ since $B_n \to 0$ in the ultra-weak topology. Since the sequence $\pi (B_n)$ is norm bounded (because $\pi$ is norm decreasing) and $\hat{\hh}_0 / \rad{\scal{\cdot}{\cdot}_0}$ is dense in $\hat{\hh}$, this suffices to show normality of $\pi$.

Since $\pi$ is a normal $\ast$-homomorphism of $\lh$, there exists a Hilbert space $\kk$ such that $\hat{\hh} = \kk \otimes \hh$ and $\pi(A) = I \otimes A$ for all $A\in\lh$; see Lemma 9.2.2 in \cite{QTOS76}. The separability of $\kk$ follows directly from the separability of $\hat{\hh}$. Since $\pi$ commutes with $\hat{P}$ and $\hat{D}$, there exist a projection valued measure $P : \bor{\Omega} \frecc \elle{\kk}$ and a strongly continuous unitary representation $D : G \frecc \elle{\kk}$ such that $\hat{P} = P \otimes I$ and $\hat{D} = D \otimes I$. Equation~(\ref{covar. nella dim.}) then implies that condition (\ref{eq:struttura:P}) holds. 

\item[(F)] We define the following operator $W : \hh \frecc \hat{\hh}_0$
\begin{equation*}
Wv = 1 \hotimes I \hotimes v \quad \forall v \in \hh.
\end{equation*}
We have $\scal{Wv}{Wv}_0 = \tr{\ii_\Omega (v\odot v)} = \norm{v}^2$, so $W$ descends to an isometry $W : \hh \frecc \hat{\hh}$. Clearly, $W U_g v = \hat{D}_g \pi(U_g) Wv$, so that condition (\ref{eq:struttura:W}) holds.

For all $A\in \lh$, we get
\begin{eqnarray*}
&& \trped{\hh}{A\, \trap{\kk}{(P (X) \otimes I) W (u \odot v) W^\ast}} =
\trped{\hat{\hh}}{(P (X) \otimes A) W (u \odot v) W^\ast} \\
&&\qquad \qquad =\scal{\hat{P} (X) \pi (A) Wu}{Wv}_{\hat{\hh}} = \trped{\hh}{A \ii_X (u\odot v)},
\end{eqnarray*}
and hence $\ii_X (u\odot v) = \trap{\kk}{(P (X) \otimes I) W (u \odot v) W^\ast}$. By the continuity of $\ii_X$ formula (\ref{eq:struttura:I}) holds for every $T\in\trh$. 

The set $\left\{ \hat{P} (X) \pi(A) Wv \mid X\in\bor{\Omega},\, A\in\lh,\, v\in\hh \right\}$ spans $\hat{\hh}_0$ and hence the set expressed in (\ref{eq:struttura:set}) is total in $\kk\otimes\hh$.

\item[(G)] Finally, we prove the last claim of Theorem \ref{teo:struttura}. Suppose $\kk^\prime$, $P^\prime$, $D^\prime$, $W^\prime$ are as stated in the theorem. Let $\hat{V} : \hat{\hh}_0 \frecc \kk^\prime \otimes \hh$ be the linear operator whose action on decomposable elements is
\begin{equation*}
\hat{V} \lft f \hotimes A \hotimes v \rgt = \lft P^\prime (f) \otimes A \rgt W^\prime v,
\end{equation*}
where we set $P^\prime (f) = \int f(x) \de P^\prime (x)$.

For an element $\phi = \sum_{i=1}^n \chi_{X_i} \hotimes A_i \hotimes v_i$ in $\hat{\hh}_0$, we have
\begin{eqnarray*}
\scal{\hat{V}\phi}{\hat{V}\phi}_0 & = & \sum_{i,j=1}^n \scal{\lft P^\prime (X_i) \otimes A_i \rgt W^\prime v_i}{\lft P^\prime (X_j) \otimes A_j \rgt W^\prime v_j}_{\kk^\prime \otimes \hh} \\
& = & \sum_{i,j=1}^n \trped{\kk^\prime \otimes \hh}{\lft P^\prime (X_i \cap X_j) \otimes A_j^\ast A_i \rgt W^\prime (v_i \odot v_j) W^{\prime \ast}} \\
& = & \sum_{i,j=1}^n \trped{\hh}{A_j^\ast A_i \ii_{X_i \cap X_j} (v_i \odot v_j)} = \scal{\phi}{\phi}_0.
\end{eqnarray*}
Hence, $\hat{V}$ descends to an isometry from $\hat{\hh}$ to $\kk^\prime \otimes \hh$. Since its image is dense in $\kk^\prime \otimes \hh$, $\hat{V}$ is actually unitary.

We have
\begin{equation}\label{uno}
\hat{V} \pi (B) ( f \hotimes A \hotimes v ) =
( P^\prime (f) \otimes BA ) W^\prime v = ( I \otimes B ) \hat{V} ( f \hotimes A \hotimes v ),
\end{equation}
and
\begin{equation}\label{due}
\hat{V} \hat{P} (X) ( f \hotimes A \hotimes v ) =
( P^\prime (\chi_X f) \otimes A ) W^\prime v = ( P^\prime (X) \otimes I ) \hat{V} ( f \hotimes A \hotimes v ).
\end{equation}
Moreover
\begin{eqnarray}
\hat{V} \hat{D}_g ( f \hotimes A \hotimes v ) & = & \hat{V} ( g\cdot f \hotimes A U_g^\ast \hotimes U_g v ) = \lft P^\prime (g\cdot f) \otimes A U_g^\ast \rgt W^\prime U_g v \notag \\
& = & \lft D^\prime_g P^\prime (f) D^{\prime\ast}_g \otimes A U_g^\ast \rgt \lft D^\prime_g \otimes U_g \rgt W^\prime v \notag \\
& = & \lft D^\prime_g P^\prime (f) \otimes A \rgt W^\prime v = \lft D^\prime_g \otimes I \rgt \hat{V} ( f \hotimes A \hotimes v ). \label{tre}
\end{eqnarray}
From eq.~(\ref{uno}) it follows that $\hat{V} (I \otimes B) = (I \otimes B) \hat{V}$ for all $B\in\lh$, hence $\hat{V} = V \otimes I$ for some unitary operator $V : \kk\frecc\kk^\prime$. From eq.~(\ref{due}) we get $\hat{V} (P(X) \otimes I) = (P^\prime (X) \otimes I) \hat{V}$, from which it follows $V P(X) = P^\prime (X) V$ for all Borel sets $X$. Finally, $\hat{V} (D_g \otimes I) = (D^\prime_g \otimes I) \hat{V}$ by eq.~(\ref{tre}), so that $V D_g = D^\prime_g V$ for all $g$.
\end{itemize}
\end{proof}

Theorem \ref{teo:struttura} can be written in an alternative form which has a more direct physical interpretation. We recall that a \emph{measurement model} $\memo$ is a 4-tuple $<\hh_\aa,Z,\xi,V>$ where 
\begin{itemize}
\item $\hh_\aa$ is a Hilbert space associated to a measurement apparatus $\aa$;  
\item $Z:\boro\to\mathcal{L}(\hh_\aa)$ is a projection valued measure (\emph{pointer observable});
\item $T_\xi$ is a one-dimensional projection corresponding to a unit vector $\xi\in\hh_\aa$ (\emph{initial state of $\aa$});
\item $V$ is a unitary operator on $\hh_\aa\otimes\hh$ (\emph{measurement coupling}).
\end{itemize}
The measurement model $\memo$ determines an instrument $\ii^{\memo}$ through the formula
\begin{equation*}
\ii^{\memo}_X \left( T\right)  =  \trap{\hh_{\aa}}{V\left(T_\xi\otimes T\right) V^\ast (Z(X)\otimes I)},\quad X\in\bor{\Omega},\, T\in\trh.
\end{equation*}

A measurement model formalizes the idea that the system is made to interact with a measurement apparatus and then a pointer observable of the apparatus is measured. This is done in order to achieve some information about the system or to prepare it in some way. The corresponding instrument gives the total description of the measurement procedure from the point of view of the system. 

Ozawa proved in \cite{Ozawa84} that for each instrument $\ii$ there is a measurement model $\memo$ such that $\ii=\ii^\memo$. In other words, all instruments arise from measurement models. The following corollary of Theorem \ref{teo:struttura} is a covariant generalization of this result and the proof follows the proof given by Ozawa. 

\begin{corollary}\label{cor:structure}
Let $\ii$ be a $U$-covariant instrument. Then there are a measurement model $\memo = <\hh_\aa,Z,\xi,V>$ and a strongly continuous unitary representation $g \mapsto R_g$ of $G$ on $\mathcal{H_A}$ such that $\ii=\ii^\memo$ and the pointer observable $Z$ satisfies the covariance condition
\begin{equation*}
R_gZ(X)R_g^\ast=Z(g\cdot X)\quad \forall X\in\bor{\Omega},\, g\in G.
\end{equation*}
\end{corollary}

\begin{proof}
With the notations of Theorem \ref{teo:struttura}, we denote $\hh_\aa=\kk\otimes\hh\otimes\kk$, $Z = I \otimes I \otimes P$, and $R = I\otimes I \otimes D$. We fix unit vectors $\xi^\prime\in\kk, \xi^{\prime\prime}\in\kk\otimes\hh$ and denote by $[\xi^\prime], [\xi^{\prime\prime}]$ the one-dimensional subspaces they generate. Then we define a mapping $V^\prime$ from $[\xi^{\prime\prime}] \otimes [\xi^\prime] \otimes \hh$ into $\kk\otimes\hh\otimes\kk\otimes \hh$ by
\begin{equation*}
V^\prime(\xi^{\prime\prime} \otimes \xi^\prime \otimes \psi) = \xi^{\prime\prime} \otimes W\psi .
\end{equation*}
The mapping $V^\prime$ is an isometry and it has a unitary extension $V$ on $\kk\otimes\hh\otimes\kk\otimes \hh$. Choosing $\xi=\xi^{\prime\prime}\otimes \xi^\prime$ we get a measurement model with the required properties.
\end{proof}

\section{The case of an irreducible representation}\label{sec:irriducibile}

In this section we make the following assumptions:
\begin{itemize}
\item $U$ is an irreducible representation of $G$;
\item $\Omega$ is the quotient space $G/H$, where $H$ is a compact subgroup of $G$. 
\end{itemize} 

We denote the left $H$-coset of $g\in G$ by $\dot{g}$. Let $\mu_G$ be a left invariant Haar measure on $G$ and let $\Delta$ denote the modular function of $G$. As the subgroup $H$ is compact, it has a Haar measure $\mu_H$ with $\mu_H (H) = 1$. Finally, $\mu_\Omega$ is the $G$-invariant measure on $\Omega$ satisfying
\begin{equation*}
\int_G f(g) \de\mu_G (g) = \int_\Omega \de\mu_\Omega (\dot{g}) \int_H f(gh) \de\mu_H (h)
\end{equation*}
for all compactly supported continuous functions $f$ on $G$.
 
We recall that the representation $U$ is called \emph{square integrable} if there exists a nonzero vector $v\in\hh$ such that the map $g\mapsto \scal{v}{U_g v}$ is in $\ldueG$. We denote by $L$ the left regular representation of $G$ acting in $\ldueG$. We will need the following result of Duflo and Moore \cite{DuMo76}.

\begin{theorem}\label{teo:square}
The representation $U$ is square integrable if and only if it is a subrepresentation of the left regular representation. In this case, there exists a unique selfadjoint
injective positive operator $C$ with $U$-invariant domain such that the following conditions hold:
\begin{enumerate}
\item for all $g\in G$,
\begin{equation*}
U_g C = \Delta ( g )^{-1/2} C U_g ;
\end{equation*}
\item for all $u\in \hh$ and $v\in\dom C$ \label{1 di Teo. quadrato integrabile}
\begin{equation*}
\int_G \Delta(g)^{-1} \left| \scal{Cv}{U_g u} \right|^2 \de\mu_G (g) = \norm{v}^2 \norm{u}^2 ;
\end{equation*}
\item if $W : \hh\frecc \ldueG$ is a bounded map intertwining $U$ with $L$, then there exists a unique $u\in\hh$ such that \label{2 di Teo. quadrato integrabile}
\begin{equation*}
Wv (g) = \Delta(g)^{-1/2} \scal{Cv}{U_g u} \quad \forall v\in\dom C.
\end{equation*}
\end{enumerate}
\end{theorem}

The square $C^2$ of the operator $C$ is called the {\em formal degree} of $U$ with respect to the Haar measure $\mu_G$. If $G$ is unimodular, then $\dom C = \hh$ and $C$ is a scalar multiple of the identity operator of $\hh$.

Let $\vv$ be a separable Hilbert space. The tensor product $\ldueG \otimes \vv $ is identified with $\ldueGV$ in the usual way. We also use the canonical identification of the tensor product $\kk \otimes \hh^\ast$ with the Hilbert space of the Hilbert-Schmidt operators of $\elle{\hh ; \kk}$. We have the following consequence of Theorem \ref{teo:square}.

\begin{corollary}\label{cor:square:1}
Suppose there exists an isometry $W : \hh\frecc\ldueG \otimes \vv $ intertwining $U$ with $L\otimes I$. Then $U$ is square integrable. Moreover, if $C$ is as in Theorem \ref{teo:square}, then there exists $B\in\vv \otimes \hh^\ast$ with $\norm{B} = 1$ such that
\begin{equation*}
Wv (g) = \Delta (g)^{-1/2} B U_g^\ast C v \quad \forall v\in\dom C.
\end{equation*}
In particular, Wv is a continuous function in $\ldueGV$ for all $v\in\dom C$.
\end{corollary}

\begin{proof}
Fix an orthonormal basis $\{ e_i \}_{i\in\nat}$ of $\vv$. Let $P_i : \ldueGV \frecc \ldueG$ and $Q_i : \ldueG \frecc \ldueGV$ be the following bounded maps
\begin{eqnarray*}
P_i f (g) & = & \scal{f(g)}{e_i} \quad \forall f\in\ldueGV \\
Q_i f (g) & = & f(g) e_i \quad \forall f\in\ldueG .
\end{eqnarray*}
Clearly, $Q_i$ is an isometry, $Q_i P_i$ is a projection operator in $\ldueGV$, $Q_i P_i Q_j P_j = 0$ if $i\neq j$, and $\sum_i Q_i P_i = I$ (in the strong operator topology). Moreover, $P_i W U_g = L_g P_i W$, hence, by item (\ref{2 di Teo. quadrato integrabile}) of Theorem \ref{teo:square}, there exists $u_i \in \hh$ such that
\begin{equation*}
( P_i W v ) (g) = \Delta(g)^{-1/2} \scal{Cv}{U_g u_i} \quad \forall v\in\dom C.
\end{equation*}
For any $v\in\dom C$, we have
\begin{eqnarray*}
\norm{v}^2 & = & \norm{W v}^2 = \sum\nolimits_{i} \norm{Q_i P_i W v}^2 = \sum\nolimits_{i} \norm{P_i W v}^2 \\
& = & \sum\nolimits_{i} \norm{v}^2 \norm{u_i}^2,
\end{eqnarray*}
the last equality following from item (\ref{1 di Teo. quadrato integrabile}) of Theorem \ref{teo:square}. Therefore, $\sum\nolimits_{i} \norm{u_i}^2 = 1$. If $u_i \neq 0$, then $\norm{u_i}^{-1} P_i W$ is an isometry intertwining $U$ with $L$, so $U$ is a subrepresentation of $L$. Thus, $U$ is square integrable by Theorem \ref{teo:square}. Moreover, the sum $\sum\nolimits_{i} e_i \odot u_i$ converges in $\vv \otimes \hh^\ast$ to an operator $B$ with $\norm{B} = 1$. If $v\in\dom C$, we have for all $g$
\begin{equation*}
\sum\nolimits_{i} ( Q_i P_i W v ) (g) = \sum\nolimits_{i} \Delta(g)^{-1/2} \scal{Cv}{U_g u_i} e_i =
\Delta(g)^{-1/2} BU_g^\ast C v.
\end{equation*}
Since $Wv = \sum\nolimits_{i} Q_i P_i W v$ (convergence in $\ldueGV$), by uniqueness of the limit
\begin{equation*}
Wv (g) = \Delta(g)^{-1/2} BU_g^\ast C v.
\end{equation*}
\end{proof}

We briefly recall some basic facts about induced representations and imprimitivity systems \cite{Mackey52}. Suppose $\sigma$ is a strongly continuous unitary representation of $H$ in $\vv$. For $f\in\ldueGV$, we define
\begin{equation*}
[\Pi f] (g) = \int_H \sigma_h f(gh) \de\mu_H (h) . 
\end{equation*}
Then $\Pi$ is a projection operator in $\ldueGV$ and it commutes with the operator $L\otimes I$. We denote by $\hh^\sigma$ the range of $\Pi$, and by $L^\sigma$ the restriction of $L \otimes I$ to $\hh^\sigma$. Observe that $\Pi f$ is a continuous function if $f$ is continuous.

For every $X\in\boro$ and $f\in\ldueGV$, we define
\begin{equation*}
[P (X) f ] (g) = \chi_X (\dot{g}) f(g). 
\end{equation*}
Then $P (X)$ is a projection operator in $\ldueGV$ commuting with $\Pi$, the map $P : \boro \frecc \elle{\ldueGV}$ is a projection valued measure, and
\begin{equation*}
L_g P (X) L_g^\ast = P (g\cdot X) \quad \forall X\in\boro,\, g\in G.
\end{equation*}
We denote by $P^\sigma$ the restriction of $P$ to $\hh^\sigma$. The triple $( L^\sigma, P^\sigma, \hh^\sigma)$ is the \emph{imprimitivity system induced by $\sigma$}.

\begin{corollary}\label{cor:square:2}
Suppose there exists an isometry $W : \hh\frecc\hh^\sigma$ intertwining $U$ with $L^\sigma$. Then $U$ is square integrable. Moreover, if $C$ is as in Theorem \ref{teo:square}, there exists $B\in\vv \otimes \hh^\ast$ with $\norm{B} = 1$ and $B U_h = \sigma_h B$ for all $h\in H$, such that
\begin{equation}\label{forma di W in accasigma}
Wv (g) = \Delta (g)^{-1/2} B U_g^\ast C v \quad \forall v\in\dom C .
\end{equation}
In particular, $Wv$ is a continuous function in $\hh^\sigma$ for all $v\in\dom C$.

Conversely, suppose $U$ is square integrable, and $B\in\vv \otimes \hh^\ast$ is such that $\norm{B} = 1$ and $B U_h = \sigma_h B$ for all $h\in H$. Then, for $v\in \dom C$, $Wv$ given by eq.~(\ref{forma di W in accasigma}) is a function in $\hh^\sigma$, and $W$ extends to an isometry from $\hh$ into $\hh^\sigma$ which intertwines $U$ with $L^\sigma$.
\end{corollary}

\begin{proof}
Suppose $W : \hh\frecc\hh^\sigma$ intertwines $U$ with $L^\sigma$. By Corollary \ref{cor:square:1}, $U$ is square integrable and eq.~(\ref{forma di W in accasigma}) holds for some $B\in\vv \otimes \hh^\ast$ with $\norm{B} = 1$. If $v\in\dom C$, then $Wv$ is a continuous function, hence $\Pi Wv$ is also continuous, and we can evaluate $Wv$ and $\Pi Wv$ at the identity $e$ of $G$. $W = \Pi W$ then gives
\begin{equation*}
B C v = Wv (e) \equiv (\Pi Wv) (e) = \int_H \sigma_h B U_h^\ast C v \de\mu_H (h).
\end{equation*}
Since the range of $C$ is dense in $\hh$, this implies
\begin{equation*}
B = \int_H \sigma_h B U_h^\ast \de\mu_H (h)
\end{equation*}
(in the strong sense), which is equivalent to $B U_h = \sigma_h B$ for all $h\in H$.

Conversely, suppose $U$ is square integrable, and let $B$ be a norm $1$ element in $\vv \otimes \hh^\ast$ intertwining $U|_H$ with $\sigma$. For $v\in \dom C$, let $Wv$ be as in eq.~(\ref{forma di W in accasigma}). If $B^\ast B = \sum\nolimits_i \lambda_i h_i \odot h_i$ is the spectral decomposition of $B^\ast B$, with $\sum_i \lambda_i = 1$, we have
\begin{eqnarray*}
\int_G \norm{Wv (g)}^2 \de\mu_G (g)
& = & \int_G \Delta (g)^{-1} \scal{B^\ast B U^\ast_g C v}{U^\ast_g C v} \de\mu_G (g) \\
& = & \sum\nolimits_i \lambda_i \int_G \Delta (g)^{-1} \left| \scal{C v}{U_g h_i} \right|^2 \de\mu_G (g) = \norm{v}^2
\end{eqnarray*}
by Theorem \ref{teo:square}. This shows that $W$ extends to an isometry from $\hh$ into $\ldueGV$. Since $\Pi Wv = Wv$ for $v\in \dom C$, $W$ maps $\hh$ into $\hh^\sigma$. Finally, the intertwining property is immediate by eq.~(\ref{forma di W in accasigma}).
\end{proof}

From now on, we fix a representation $\sigma$ of $H$ with the following property: if $\sigma^\prime$ is another strongly continuous unitary representation of $H$ acting in a separable Hilbert space $\vv^\prime$, then $\sigma^\prime$ is a subrepresentation of $\sigma$.\footnote{Since $H$ is compact and separable, there exists a representation $\sigma$ having such property, and $\sigma$ is unique up to unitary equivalence. An explicit realization of $\sigma$ is obtained in this way: let $\kk$ be a separable infinite dimensional Hilbert space and $L^H$ the regular representation of $H$ acting in $L^2 (H,\mu_H)$. Then the representation $L^H \otimes I$ acting in $L^2 (H,\mu_H) \otimes \kk$ contains every irreducible representation of $H$ with infinite multiplicity and hence has the required property.} We define the following set associated to $U$ and $\sigma$
\begin{equation*}
\cc := \left\{ B\in \vv \otimes \hh \otimes \hh^\ast \mid \norm{B} = 1 \textrm{ and } B U_h = ( \sigma_h \otimes U_h ) B \textrm{ for all } h\in H \right\}
\end{equation*}

Suppose $U$ is square integrable. We denote by $\trhzero$ the following linear subspace of $\trh$ 
\begin{equation*}
\trhzero := \textrm{span}\left\{ u\odot v \mid u,v \in \dom C \right\}.
\end{equation*}
Since the domain of $C$ is dense in $\hh$, the set $\trhzero$ is dense in $\trh$.
For $B\in\cc$ and $T = \sum_{i = 1}^n v_i\odot u_i\in\trhzero$, $u_i, v_i \in\dom C$, we define
\begin{equation*}
\phi_B (T,g) := \Delta(g)^{-1} \sum_{i = 1}^n \trap{\vv}{ ( I_{\vv} \otimes U_g ) B U^\ast_g C v_i \odot ( I_{\vv} \otimes U_g ) B U^\ast_g C u_i } .
\end{equation*}
The map $g\mapsto \phi_B (T,g)$ is continuous from $G$ into $\trh$, and it is constant on the left $H$-cosets. Thus, it descends to a continuous function of $\Omega$ into $\trh$. Moreover, for all $\dot{g}$, the map $T\mapsto \phi_B (T,\dot{g})$ is linear and positive from $\trhzero$ into $\trh$.

We claim that $\phi_B (T,\cdot) \in L^1 (\Omega , \mu_{\Omega} ; \trh)$. In fact, if $T = \sum_{i = 1}^n v_i \odot v_i$, with $v_i \in \dom C$, is a positive element of $\trhzero$, we have
\begin{eqnarray*}
\int_{\Omega} \norm{\phi_B (T,\dot{g})}_\mathcal{T} \de\mu_{\Omega} (\dot{g})
& = & \int_{G} \trped{\hh}{\phi_B (T,g)} \de\mu_{G} (g) \\
& = & \sum_{i = 1}^n \int_{G} \Delta(g)^{-1} \scal{ B^\ast B U^\ast_g C v_i }{U^\ast_g C v_i } \de\mu_{G} (g).
\end{eqnarray*}
Let $B^\ast B = \sum\nolimits_j \lambda_j h_j \odot h_j$ be the spectral decomposition of $B^\ast B$, with $\sum_j \lambda_j = 1$. By Theorem \ref{teo:square} we get
\begin{eqnarray*}
\int_{\Omega} \norm{\phi_B (T,\dot{g})}_\mathcal{T} \de\mu_{\Omega} (\dot{g})
& = & \sum_{i = 1}^n \sum\nolimits_j \lambda_j \int_{G} \Delta (g)^{-1} \left| \scal{C v_i}{U_g h_i} \right|^2 \de\mu_{G} (g) \\
& = & \sum\nolimits_j \lambda_j \sum_{i = 1}^n \norm{v_i}^2 = \norm{T}_\mathcal{T} .
\end{eqnarray*}
If $T$ is generic element in $\trhzero$, decomposing it as $T = T^1_+ - T^1_- + i (T^2_+ - T^2_-)$, with $T^i_{\pm}$ positive elements in $\trhzero$ and $\norm{T^i_+}_{\ttt} + \norm{T^i_-}_{\ttt} \leq \norm{T}_{\ttt}$, we get by the above equation
\begin{equation}\label{una maggiorazione}
\int_{\Omega} \norm{\phi_B (T,\dot{g})}_\mathcal{T} \de\mu_{\Omega} (\dot{g}) \leq 2 \norm{T}_\mathcal{T} < \infty,
\end{equation}
and the claim is proved.

\begin{theorem}\label{Teo. sugli str. cov.}
Suppose $U$ is square integrable. If $B\in\cc$, there is a unique instrument $\ii^B : \trh \frecc \mm (\Omega ; \trh)$ such that for $T\in\trhzero$
\begin{equation}\label{def. di IB}
(\ii^B T) (X) = \int_X \phi_B (T,\dot{g}) \de\mu_{\Omega} (\dot{g}) \quad \forall X\in\bor{\Omega},
\end{equation}
the integral converging in the trace class norm.
The instrument $\ii^B$ is covariant with respect to $U$.

Conversely, if $\ii$ is an instrument based on $\Omega$ and covariant with respect to $U$, then $U$ is square integrable, and there exists $B\in\cc$ such that $\ii = \ii^B$.
\end{theorem}

\begin{proof}
Convergence of the integral (\ref{def. di IB}) in the trace class norm follows from eq.~(\ref{una maggiorazione}).

If $\ii$ is an instrument based on $\Omega$ and covariant with respect to $U$, by Theorem \ref{teo:struttura} there exists an  imprimitivity system $(D,P,\kk)$ based on $\Omega$ and an isometry $W : \hh \frecc \kk\otimes\hh$ intertwining $U$ with $D\otimes U$ such that $\ii T (X) = \trap{\kk}{\lft P(X) \otimes I \rgt W T W^\ast}$. By imprimitivity theorem, $(D,P,\kk)$ is the imprimitivity system induced by some representation $\sigma^\prime$.
It is not restrictive to assume that $\sigma^\prime$ is the largest possible, i.~e.~$\sigma^\prime \equiv \sigma$, so that $(D,P,\kk) = (L^\sigma, P^\sigma, \hh^\sigma)$.

Conversely, if $W: \hh\frecc\hh^\sigma \otimes \hh$ is an isometry intertwining $U$ with $L^\sigma \otimes U$, then formula $\ii T (X) = \trap{\hh^\sigma}{\lft P^\sigma (X) \otimes I \rgt W T W^\ast}$ defines a $U$-covariant instrument $\ii$ based on $\Omega$. Therefore, the problem of characterizing the $U$-covariant instruments reduces to the task of finding the most general intertwining isometry $W: \hh\frecc\hh^\sigma \otimes \hh$. 

The unitary operator $V : \ldueGV \otimes \hh \frecc \ldueGV \otimes \hh$ given by
\begin{equation*}
V f (g) = (I_{\vv} \otimes U^\ast_g) f(g) \quad \forall f\in\ldueGVhh
\end{equation*}
intertwines $L\otimes I_{\vv} \otimes U$ with $L\otimes I_{\vv} \otimes I_{\hh}$, and $V (\hh^\sigma \otimes \hh) = \hh^{\sigma\otimes U|_H}$. Hence, $VW : \hh \frecc \hh^{\sigma\otimes U|_H}$ is an isometry intertwining $U$ with $L^{\sigma\otimes U|_H}$, and Corollary \ref{cor:square:2} applies. In particular, there exist $U$-covariant instruments based on $\Omega$ if and only if $U$ is square integrable. The most general form of $W$ is thus
\begin{equation*}
W v (g) = V^\ast V W v (g) = \Delta (g)^{-1/2} (I_{\vv} \otimes U_g) B U^\ast_g C v \quad \forall v\in\dom C ,
\end{equation*}
with $B\in\cc$.

With $W$ as above, if $T = \sum_{i = 1}^n v_i\odot u_i$ with $u_i , v_i \in \dom C$, and $v\in\hh$, we have
\begin{eqnarray*}
&& \scal{\ii_X (T) v}{v} = \trped{\hh}{(v\odot v) \trap{\hh^\sigma}{(P^\sigma(X) \otimes I_{\hh}) W T W^\ast}} \\
&& \quad = \trped{\hh^\sigma \otimes \hh}{(P^\sigma(X) \otimes (v\odot v)) W T W^\ast} \\
&& \quad = \sum_{i = 1}^n \scal{(P^\sigma(X) \otimes (v\odot v)) W v_i}{W u_i}_{\hh^\sigma \otimes \hh} \\
&& \quad = \int_X \Delta (g)^{-1} \sum_{i = 1}^n \scal{\trap{\vv}{(I_{\vv} \otimes U_g) B U^\ast_g C v_i \odot (I_{\vv} \otimes U_g) B U^\ast_g C u_i} v}{v}_{\hh} \de\mu_{\Omega} (\dot{g}),
\end{eqnarray*}
i.e.~$\ii T$ is given by
\begin{equation*}
\ii T (X) = \int_X \phi_B (T,\dot{g}) \de\mu_{G}(g) \quad \forall X\in\bor{G}.
\end{equation*}
Uniqueness and covariance of $\ii^B$ then follows as $\trhzero$ is dense in $\trh$ and  $\ii$ is continuous.
\end{proof}

If $H$ is the trivial one element subgroup of $G$, then Theorem \ref{Teo. sugli str. cov.} can be written in the following simplified form.

\begin{corollary}\label{cor:H=e}
There exist $U$-covariant instruments based on $G$ if and only if $U$ is square integrable. In this case, if $B\in \hh \otimes \hh^\ast$ has norm $1$, there exists a unique instrument $\jj^{B} : \trh \frecc \mm (G ; \trh)$ such that, for $T =  v\odot u$ with $u, v \in \dom C$,
\begin{equation}\label{eq:H=e}
(\jj^{B} T)(X) = \int_X \Delta (g)^{-1} U_g B U^\ast_g C v \odot U_g B U^\ast_g C u \ \de\mu_{G} (g) \quad \forall X\in\bor{G},
\end{equation}
the integral converging in the trace class norm. The instrument $\jj^{B}$ is $U$-covariant.

If $\{ \lambda_j \}_{j\in\nat}$ is a sequence of nonnegative real numbers summing up to $1$ and $\{ B_j \}_{j\in\nat}$ is a sequence of norm $1$ elements in $\hh\otimes\hh^\ast$, then the series
$\sum_{j\in\nat} \lambda_j \jj^{B_j}$ converges absolutely in $\elle{\trh ; \mm (G ; \trh)}$ to a $U$-covariant instrument, and every $U$-covariant instrument is of the form $\sum_{j\in\nat} \lambda_j \jj^{B_j}$ with $\lambda_j \geq 0$, $\sum_{j\in\nat} \lambda_j = 1$ and $\norm{B_j} = 1$.
\end{corollary}

\begin{proof}
The first claim follows directly from Theorem \ref{Teo. sugli str. cov.}.

Suppose $U$ is square integrable. For $H = \{ e \}$, $\sigma$ is the trivial representation in an infinite dimensional Hilbert space $\vv$. If $k\in\vv$ and $B\in\hh\otimes \hh^\ast$ with $\norm{k} = \norm{B} = 1$, then $k\otimes B \in\cc$. By Theorem \ref{Teo. sugli str. cov.}, there is a unique $U$-covariant instrument $\ii^{k\otimes B}$ satisfying formula (\ref{def. di IB}). If $T = \sum_{i=1}^n v_i\odot u_i$ with $u_i, v_i \in \dom C$, we have
\begin{equation*}
(\ii^{k\otimes B} T) (X) = \int_X \Delta (g)^{-1} \sum_{i = 1}^n U_g B U^\ast_g C v_i \odot U_g B U^\ast_g C u_i \ \de\mu_{G}(g),
\end{equation*}
the integral converging in the trace class norm. By eq.~(\ref{una maggiorazione}) we have
\begin{equation}\label{*}
\int_X \Delta (g)^{-1} \norm{\sum_{i = 1}^n U_g B U^\ast_g C v_i \odot U_g B U^\ast_g C u_i}_\mathcal{T} \de\mu_{G}(g) \leq 2 \norm{T}_\mathcal{T}.
\end{equation}
We denote $\jj^B=\ii^{k\otimes B}$, and this is thus the instrument in (\ref{eq:H=e}).

We recall that $\norm{\jj^{B}}_\infty \leq 2$ by Proposition \ref{prop:bounded}. Thus, if $\lambda_j \geq 0$, $\sum_{j\in\nat} \lambda_j = 1$ and $\norm{B_j} = 1$, then the sum $\sum_{j\in\nat} \lambda_j \jj^{B_j}$ is absolutely convergent. Its limit is clearly a $U$-covariant instrument.

Conversely, if $\ii$ is $U$-covariant, then $\ii = \ii^B$ for some $B\in\cc$ by Theorem \ref{Teo. sugli str. cov.}. Let $\{ e_i \}$ be a Hilbert basis of $\vv$. Then, $B = \sum_j e_j \otimes B_j$, with $B_j \in \hh\otimes\hh^\ast$ and $\sum_j \norm{B_j}^2 = 1$. Formula~(\ref{def. di IB}) for $T = \sum_{i = 1}^n v_i \odot u_i$, with $u_i , v_i \in \dom C$, can be written as
\begin{equation*}
(\ii^{B} T)(X) = \int_X \Delta (g)^{-1} \sum_{i = 1}^n \sum\nolimits_j U_g B_j U^\ast_g C v_i \odot U_g B_j U^\ast_g C u_i \de\mu_{G} (g).
\end{equation*}
By eq.~(\ref{*}) and dominated convergence theorem we get
\begin{eqnarray*}
(\ii^{B} T)(X) & = & \sum\nolimits_j \int_X \Delta (g)^{-1} \sum_{i = 1}^n U_g B_j U^\ast_g C v_i \odot U_g B_j U^\ast_g C u_i \de\mu_{G} (g) \notag \\
& = & \sum\nolimits_j \norm{B_j}^2 (\jj^{B_j / \norm{B_j}} T) (X).
\end{eqnarray*}
Thus, $\ii^B = \sum\nolimits_j \norm{B_j}^2 \jj^{B_j / \norm{B_j}}$ as the set $\trhzero$ is dense in $\trh$.
\end{proof}

\section{Covariant instruments and projective representations}\label{sec:projective}

In this section we extend the previous results to the case in which $U$ is a projective unitary representation of $G$ in $\hh$. We recall that a projective unitary representation of $G$ in $\hh$ is a mapping $U : G\frecc \elle{\hh}$ such that
\begin{enumerate}
\item $U$ is a weakly Borel map;
\item $U(e) = I$;
\item there exists a mapping $m : G\times G \frecc \torus$ ($\torus$ being the group of complex numbers with modulus one) such that $U(g_1 , g_2) = m(g_1 , g_2) U(g_1) U(g_2)$. 
\end{enumerate}
The function $m$ is the {\em multiplier} of $U$. Also in this case, we will often use the abbreviated notation $U_g = U(g)$. 

For more details about projective representations we refer to~\cite{GQT85} and \cite{TSAQM04}. Here we recall that the set $G_m := G \times \torus$ endowed with the product law
$$
(g,z) (g^\prime , z^\prime) = (g g^\prime , z z^\prime m(g,g^\prime))
$$
is a group, and there exists a unique lcsc~topology on $G_m$ making it a topological group with $\torus$ being central closed subgroup and $G_m /\torus = G$. The group $G_m$ is called the {\em central extension} of $G$ associated to the multiplier $m$.

The projective representation $U$ extends to a strongly continuous unitary representation $\tilde{U}$ of $G_m$ by setting
\begin{equation}\label{def. di Utilde}
\tilde{U} (g,z) = z^{-1} U(g).
\end{equation}
Moreover, $U$ is irreducible if and only if $\tilde{U}$ is. The action of $G$ on $\Omega$ lifts to an action of $G_m$ on $\Omega$, with $\torus$ acting trivially.

Definition \ref{def:cov} of a covariant instrument clearly makes sense also in the case of projective representations. It is immediately checked that the instrument $\ii$ is covariant with respect to the projective representation $U$ of $G$ if and only if it is covariant with respect to the representation $\tilde{U}$ of $G_m$. Therefore, Theorem \ref{teo:struttura} is valid also in the case of projective representations. 

Suppose $U$ is an irreducible projective unitary representation of $G$. As for usual representations, we say that $U$ is square integrable if the mapping $g\mapsto \scal{v}{U_g v}$ is in $\ldueG$ for some nonzero $v$. Then $U$ is square integrable if and only if the representation $\tilde{U}$ of $G_m$ is square integrable in the usual sense. In fact, if $\mu_{\torus}$ is the normalized Haar measure of $\torus$, then $\de \mu_{G_m} (g,z) = \de \mu_G (g) \de \mu_{\torus} (z)$ is a Haar measure of $G_m = G\times \torus$, and
$$
\int_G |\scal{v}{U_g v}|^2 \de\mu_G (g) = \int_{G\times \torus} \left|\scal{v}{\tilde{U} (g,z) v} \right|^2 \de\mu_{G_m} (g,z).
$$
The formal degree of the projective representation $U$ with respect to the Haar measure $\mu_G$ is defined as the formal degree of $\tilde{U}$ with respect to the Haar measure $\mu_{G_m}$.

Let $H\subset G$ be a compact subgroup. As we did in the previous section, we let $\sigma$ acting in the Hilbert space $\vv$ be the maximal separable unitary representation of $H$, and we denote
\begin{equation}\label{eq:C_ancora}
\cc = \left\{ B\in \vv \otimes \hh \otimes \hh^\ast \mid \norm{B} = 1 \textrm{ and } B U_h = ( \sigma_h \otimes U_h ) B \textrm{ for all } h\in H \right\}.
\end{equation}
Theorem \ref{Teo. sugli str. cov.} holds also for projective representations. In fact

\begin{corollary}\label{cor: U irr. pr., H gen.}
Suppose $U$ is a square integrable projective unitary representation of $G$. If $B\in\cc$, there is a unique instrument $\ii^B : \trh \frecc \mm (\Omega ; \trh)$ such that for $T = \sum_{i=1}^n v_i \odot u_i$, $v_i , u_i \in \dom C$,
\begin{equation}\label{def. di IB proiett.}
(\ii^B T) (X) = \int_X \phi_B (T,\dot{g}) \de\mu_{\Omega} (\dot{g}) \quad \forall X\in\bor{\Omega},
\end{equation}
where
\begin{equation*}
\phi_B (T,\dot{g}) = \Delta(g)^{-1} \sum_{i = 1}^n \trap{\vv}{ ( I_{\vv} \otimes U_g ) B U^\ast_g C v_i \odot ( I_{\vv} \otimes U_g ) B U^\ast_g C u_i }
\end{equation*}
and the integral converges in the trace class norm.
The instrument $\ii^B$ is covariant with respect to $U$.

Conversely, if $\ii$ is an instrument based on $\Omega$ and covariant with respect to $U$, then $U$ is square integrable, and there exists $B\in\cc$ such that $\ii = \ii^B$.
\end{corollary}

\begin{proof}
The subset $H_m = H \times \torus$ is a closed compact subgroup of $G_m$. The spaces $G/H$ and $G_m / H_m$ are clearly identified both under the action of $G$ and under the action of $G_m$ (with $\torus$ acting trivially).

Let $\tilde{\sigma}$ acting in $\tilde{\vv}$ be the maximal separable  unitary representation of $H_m$, and define
\begin{equation}\label{L'altra def. di cc}
\tilde{\cc} = \left\{ B\in \tilde{\vv} \otimes \hh \otimes \hh^\ast \mid \norm{B} = 1 \textrm{ and } B \tilde{U}_{\tilde{h}} = ( \tilde{\sigma}_{\tilde{h}} \otimes \tilde{U}_{\tilde{h}} ) B \textrm{ for all } \tilde{h}\in H_m \right\}.
\end{equation}

By Theorem \ref{Teo. sugli str. cov.} and the previous remarks, the statement of the above theorem is true with $\vv$ replaced by $\tilde{\vv}$ and $B\in \tilde{\cc}$. 
Decompose $\tilde{\vv} = \oplus_{n\in\Z} \vv_n$, with $\tilde{\sigma} (e,z) v = z^n v$ for all $v\in\vv_n$. Each $\vv_n$ is $\tilde{\sigma}$-invariant. If $B\in\tilde{\cc}$, then
\begin{equation*}
(\tilde{\sigma} (e,z) \otimes I) B = z (\tilde{\sigma} (e,z) \otimes \tilde{U} (e,z)) B = z B \tilde{U} (e,z) = B,
\end{equation*}
i.e.~$B\in \vv_0 \otimes \hh \otimes \hh^\ast$. Since the restriction of $\tilde{\sigma}$ to $\vv_0$ is naturally identified with $\sigma$, the claim of the theorem follows.
\end{proof}

Finally, we prove the projective version of Corollary \ref{cor:H=e}.

\begin{corollary}\label{cor:H=e:pr}
There exist $U$-covariant instruments based on $G$ if and only if $U$ is square integrable. In this case, if $B\in \hh \otimes \hh^\ast$ has norm $1$, there exists a unique instrument $\jj^{B} : \trh \frecc \mm (G ; \trh)$ such that, for $T =  v\odot u$ with $u, v \in \dom C$,
\begin{equation*}
(\jj^{B} T)(X) = \int_X \Delta (g)^{-1} U_g B U^\ast_g C v \odot U_g B U^\ast_g C u \ \de\mu_{G} (g) \quad \forall X\in\bor{G},
\end{equation*}
the integral converging in the trace class norm. The instrument $\jj^{B}$ is covariant with respect to $U$.

If $\{ \lambda_j \}_{j\in\nat}$ is a sequence of nonnegative real numbers summing up to $1$ and $\{ B_j \}_{j\in\nat}$ is a sequence of norm $1$ elements in $\hh\otimes\hh^\ast$, then the series
$\sum_{j\in\nat} \lambda_j \jj^{B_j}$ converges absolutely in $\elle{\trh ; \mm (G ; \trh)}$ to a $U$-covariant instrument, and every $U$-covariant instrument is of the form $\sum_{j\in\nat} \lambda_j \jj^{B_j}$ with $\lambda_j \geq 0$, $\sum_{j\in\nat} \lambda_j = 1$ and $\norm{B_j} = 1$.
\end{corollary}
\begin{proof}
If $\ii$ is an instrument based on $G$, for all $\tilde{X}\in \bor{G_m}$ and $T\in\trh$ define
\begin{equation*}
(\tilde{\ii} T)(\tilde{X}) = \int_G \left[ \int_{\torus} \chi_{\tilde{X}} (g,z) \de \mu_{\torus} (z) \right] \de (\ii T)(g).
\end{equation*}
It is easy to check that $\tilde{\ii}$ is an instrument based on $G_m$. $\tilde{\ii}$ is $\tilde{U}$-covariant if $\ii$ is covariant with respect to $U$.

On the other hand, defining $p:G_m \frecc G$, $p(g,z) = g$, we see that
\begin{equation*}
p_\ast : \mm (G_m ; \trh) \frecc \mm (G ; \trh), \quad p_\ast (M) (X) = M(p^{-1} (X))
\end{equation*}
is a continuous positive mapping, and, if $\tilde{\ii}$ is a $\tilde{U}$-covariant instrument based on $G_m$, then $p_\ast \tilde{\ii}$ is a $U$-covariant instrument based on $G$.

It can be easily checked that the mappings $\ii \mapsto \tilde{\ii}$ and $\tilde{\ii} \mapsto p_\ast \tilde{\ii}$ are one the inverse of the other when restricted to the set of $U$- and $\tilde{U}$-covariant instruments. The claim then follows by Corollary \ref{cor:H=e}, observing that
\begin{eqnarray*}
&& \int_{p^{-1} (X)} \Delta (g,z)^{-1} \sum_{i = 1}^n \tilde{U} (g,z) B \tilde{U} (g,z)^\ast C v_i \odot \tilde{U} (g,z) B \tilde{U} (g,z)^\ast C u_i \ \de\mu_{G_m} (g,z) \\
&& \qquad = \int_X \Delta (g)^{-1} \sum_{i = 1}^n U(g) B U^\ast (g) C v_i \odot U (g) B U^\ast (g) C u_i \ \de\mu_{G} (g).
\end{eqnarray*}
\end{proof}

\section{Covariant instruments and completely positive maps}\label{sec:cp}

In \cite{Davies70} and \cite[Sec. 4.5]{QTOS76} Davies derives a characterization for $U$-covariant instruments\footnote{In \cite{Davies70} and \cite[Sec. 4.5]{QTOS76} Davies does not require instruments to be completely positive, but he discuss this condition in \cite[Sec. 9.2.]{QTOS76}} in the case that $U$ is a finite dimensional unitary representation of a compact group $G$. His characterization is based on certain kind of positive linear maps on $\trh$. 

Assuming that $U$ is a square integrable projective unitary representation of $G$ and $H\subset G$ is a compact subgroup, we apply Corollary \ref{cor: U irr. pr., H gen.} in order to give an alternative description of the $U$-covariant instruments based on $\Omega=G/H$. This characterization is similar to that of Davies.

We denote by $\pp$ the convex set of maps $\Phi : \elle{\hh} \frecc \elle{\hh}$ such that
\begin{enumerate}
\item $\Phi$ is normal and completely positive;
\item $\Phi (I) \in \trh$ and $\tr{\Phi (I)} = 1$;
\item $\Phi (U_h A U_h^\ast ) = U_h \Phi (A) U_h^\ast$ for all $A\in\elle{\hh}$ and $h\in H$.
\end{enumerate}

\begin{lemma}\label{lemma:PhiB}
Let $\cc$ be the set defined in eq.~(\ref{eq:C_ancora}). For each $B\in\cc$, formula
\begin{equation}\label{eq:PhiB}
\Phi_B (A) = B^\ast (I_\vv \otimes A) B \quad \forall A\in\lh
\end{equation}
defines a map $\Phi_B \in\pp$. Conversely, if $\Phi\in\pp$, then there is $B\in\cc$ such that $\Phi = \Phi_B$. 
\end{lemma}

\begin{proof}
It is clear that formula~(\ref{eq:PhiB}) defines an element $\Phi_B \in \pp$.

Conversely, suppose $\Phi\in\pp$. Since $\Phi$ is normal and completely positive, there is a Hilbert space $\vv_0$ and a bounded linear map $V : \hh\frecc\vv_0 \otimes \hh$ such that
\begin{equation*}
\Phi (A) = V^\ast (I \otimes A) V \quad \forall A\in\lh
\end{equation*}
and
\begin{equation*}
\vv_0 \otimes \hh = \overline{\rm span} \, \{(I \otimes A) V v \mid A\in\lh ,\, v\in\hh \}.
\end{equation*}
By condition (2), $V\in\vv_0 \otimes \hh\otimes\hh^\ast$, and $\norm{V} = 1$.
For $h\in H$, define the following linear operator $\tilde{\sigma}_h$ in ${\rm span} \, \{(I \otimes A) V v \mid A\in\lh ,\, v\in\hh \}$
\begin{equation*}
\tilde{\sigma}_h \left[ \sum \nolimits_{i = 1}^n (I\otimes A_i) V v_i \right] = \sum\nolimits_{i = 1}^n (I\otimes A_i U_h^\ast) V U_h v_i .
\end{equation*}
$\tilde{\sigma}_h$ is well defined and isometric, since
\begin{eqnarray*}
\norm{\sum\nolimits_{i = 1}^n (I\otimes A_i U_h^\ast) V U_h v_i}^2 & = & \sum\nolimits_{i,j = 1}^n \scal{\Phi (U_h A_j^\ast A_i U_h^\ast) U_h v_i}{U_h v_j} \\
& = & \sum\nolimits_{i,j = 1}^n \scal{\Phi (A_j^\ast A_i) v_i}{v_j} \\
& = & \norm{\sum\nolimits_{i = 1}^n (I\otimes A_i ) V v_i}^2 .
\end{eqnarray*}
So, $\tilde{\sigma}_h$ extends to an isometry in $\vv_0 \otimes \hh$. It is easy to check that $\tilde{\sigma}$ is a weakly (hence strongly) continuous unitary representation of $H$ in $\vv_0 \otimes \hh$. Since $\tilde{\sigma}_h (I\otimes A) = (I\otimes A) \tilde{\sigma}_h$ for all $A\in\lh$, $\tilde{\sigma}_h = \sigma^\prime_h \otimes I$ for some representation $\sigma^\prime$ of $H$ in $\vv_0$. Let $J :\vv_0 \frecc \vv$ be an isometry intertwining $\sigma^\prime$ with $\sigma$. Equation (\ref{eq:PhiB}) then follows by setting $B = (J\otimes I) V$.
\end{proof}

If $\phi : \trh\frecc\ttt(\hh)$ is a bounded linear map, we let $\phi^\ast : \elle{\hh} \frecc \elle{\hh}$ be its adjoint.

\begin{corollary}
There is a one-to-one convex mapping $\Phi\mapsto \ii^\Phi$ of $\pp$ onto the set of $U$-covariant instrument based on $\Omega$. If $\Phi\in\pp$, then $\ii^\Phi$ is defined by
\begin{equation*}
\scal{(\ii_X^{\Phi})^\ast (A) v}{u} = \int_X \Delta(g)^{-1} \scal{U_g \Phi (U_g^\ast A U_g) U_g^\ast C v}{C u} \de\mu_{\Omega} (\dot{g}) \quad \forall v,u\in\dom C
\end{equation*}
for all $A\in\elle{\hh}$ and $X\in\bor{\Omega}$.
\end{corollary}

\begin{proof}
By Lemma \ref{lemma:PhiB}, the elements in $\pp$ are all the maps of the form $\Phi_B$ for some $B\in\cc$. For $v,u\in\dom C$, we get
\begin{eqnarray*}
&&\scal{(\ii_X^{\Phi_B})^\ast (A) v}{u} =
\int_X \Delta(g)^{-1} \scal{ U_g B^\ast ( I \otimes U_g^\ast A U_g ) B U^\ast_g C v}{C u } \de\mu_{\Omega} (\dot{g}) \\
&& \qquad\quad = \int_X \Delta(g)^{-1} \trped{\vv\otimes\hh}{ ( I \otimes A U_g ) B U^\ast_g C v \odot ( I \otimes U_g ) B U^\ast_g C u } \de\mu_{\Omega} (\dot{g}) \\
&& \qquad\quad = \tr{A\, \ii_X^B (v\odot u)} = \scal{(\ii_X^B)^\ast (A) v}{u}.
\end{eqnarray*}
This means that $\ii^{\Phi_B} = \ii^B $. By Corollary \ref{cor: U irr. pr., H gen.}, the correspondence $\Phi \mapsto \ii^\Phi$ is onto. To show the injectivity of this correspondence, suppose $\Phi , \Psi \in\pp$ are such that
$$
\scal{(\ii^{\Phi}_X)^\ast (A) v}{u} = \scal{(\ii^{\Psi}_X)^\ast (A) v}{u}
$$
for all $A\in\lh$, $X\in\bor{\Omega}$ and $v,u \in\dom C$. Then,
$$
\scal{U_g \Phi (U_g^\ast A U_g) U_g^\ast C v}{C u} = \scal{U_g \Psi (U_g^\ast A U_g) U_g^\ast C v}{C u} \quad \forall g\in G , \, v,u\in\dom C ,
$$
so, by the density of ${\rm ran}\, C$, we get
$$
U_g \Phi (U_g^\ast A U_g) U_g^\ast = U_g \Psi (U_g^\ast A U_g) U_g^\ast \quad \forall g\in G .
$$
Taking $g=e$, we get $\Phi (A) = \Psi (A)$ for all $A\in\lh$ and hence, $\Phi = \Psi$.
\end{proof}

\section*{Acknowledgements}

T.H. acknowledges the support of the European Union project CONQUEST.


\begin{thebibliography}{10}

\bibitem{QTM96}
P.~Busch, P.J. Lahti, and P.~Mittelstaedt.
\newblock {\em The Quantum Theory of Measurement}.
\newblock Springer-Verlag, Berlin, second revised edition, 1996.

\bibitem{TSAQM04}
G.~Cassinelli, E.~De~Vito, P.~Lahti, and A.~Levrero.
\newblock {\em The Theory of Symmetry Actions in Quantum Mechanics}.
\newblock Springer, Berlin, 2004.

\bibitem{Davies70}
E.B. Davies.
\newblock On the repeated measurements of continuous observables in quantum
  mechanics.
\newblock {\em J. Funct. Anal.}, 6:318--346, 1970.

\bibitem{QTOS76}
E.B. Davies.
\newblock {\em Quantum Theory of Open Systems}.
\newblock Academic Press, London, 1976.

\bibitem{DaLe70}
E.B. Davies and J.T. Lewis.
\newblock An operational approach to quantum probability.
\newblock {\em Comm. Math. Phys.}, 17:239--260, 1970.

\bibitem{Denisov90}
L.V. Denisov.
\newblock On the Stinespring-type theorem for covariant instruments in noncommutative probability theory.
\newblock {\em Reports on mathematics and its applications of the Steklov Mathematical Institute.}, Preprint n. 20, Moscow, 1990.

\bibitem{DuMo76}
M.~Duflo and C.~Moore.
\newblock On the regular representation of a nonunimodular locally compact
  group.
\newblock {\em J. Functional Analysis}, 21:209--243, 1976.

\bibitem{MT50}
P.R. Halmos.
\newblock {\em Measure {T}heory}.
\newblock D. Van Nostrand Company, Inc., New York, 1950.

\bibitem{AHAI63}
E.~Hewitt and K.A. Ross.
\newblock {\em Abstract Harmonic Analysis. {V}ol. {I}: {S}tructure of
  Topological Groups. {I}ntegration Theory, Group Representations}.
\newblock Academic Press, New York, 1963.

\bibitem{PSAQT82}
A.S. Holevo.
\newblock {\em Probabilistic and Statistical Aspects of Quantum Theory}.
\newblock North-Holland Publishing Co., Amsterdam, 1982.

\bibitem{Holevo98}
A.S. Holevo.
\newblock Radon-{N}ikod\'ym derivatives of quantum instruments.
\newblock {\em J. Math. Phys.}, 39:1373--1387, 1998.

\bibitem{Keyl02}
M.~Keyl.
\newblock Fundamentals of quantum information theory.
\newblock {\em Phys. Rep.}, 369:431--548, 2002.

\bibitem{RFA93}
S.~Lang.
\newblock {\em Real and functional analysis}, volume 142 of {\em Graduate Texts
  in Mathematics}.
\newblock Springer-Verlag, New York, third edition, 1993.

\bibitem{Mackey52}
G.W. Mackey.
\newblock Induced representations of locally compact groups. {I}.
\newblock {\em Ann. of Math. (2)}, 55:101--139, 1952.

\bibitem{Ozawa84}
M.~Ozawa.
\newblock Quantum measuring processes of continuous observables.
\newblock {\em J. Math. Phys.}, 25:79--87, 1984.

\bibitem{GQT85}
V.S. Varadarajan.
\newblock {\em Geometry of Quantum Theory}.
\newblock Springer-Verlag, New York, second edition, 1985.

\bibitem{Werner01}
R.F. Werner.
\newblock Quantum information theory -- an invitation.
\newblock In {\em Quantum Information: an Introduction to Basic Theoretical
  Concepts and Experiments}, chapter~2, pages 14--57. Springer-Verlag, 2001.



\end{thebibliography}

\end{document}